\documentclass{IEEEtran}
\usepackage{amsmath, amssymb, amsthm, enumerate}
\usepackage{cite}
\usepackage{hyperref}

\long\def\symbolfootnote[#1]#2{\begingroup
\def\thefootnote{\fnsymbol{footnote}}\footnote[#1]{#2}\endgroup}

\title{Multichannel Conflict-Avoiding Codes \\ of Weights Three and Four}

\long\def\symbolfootnote[#1]#2{\begingroup
\def\thefootnote{\fnsymbol{footnote}}\footnote[#1]{#2}\endgroup}

\ifdefined\SingleColumn
\author{Yuan-Hsun Lo$^\dagger$, Kenneth W. Shum$^\ddagger$, Wing Shing Wong$^\P$ and Yijin Zhang$^\S$ \\
$^\dagger$Department of Applied Mathematics, National Pingtung University \\
$^\ddagger$School of Science and Engineering, \\The Chinese University of Hong Kong, Shenzhen\\
$^\P$Department of Information Engineering, \\The Chinese University of Hong Kong\\
$^\S$School of Electronic and Optical Engineering, \\Nanjing University of Science and Technology
}
\else
\author{
Yuan-Hsun~Lo,~\IEEEmembership{Member,~IEEE,} Kenneth~W.~Shum,~\IEEEmembership{Senior Member,~IEEE,} Wing~Shing~Wong,~\IEEEmembership{Life~Fellow,~IEEE,} and Yijin~Zhang,~\IEEEmembership{Senior~Member,~IEEE}\\
\ \\
\thanks{
This research is partially supported by the Ministry of Science and Technology, Taiwan under grant number 108-2115-M-153-004-MY2, University Development Fund --- Research Start-up Fund (UDF01000918) from The Chinese University of Hong Kong, Shenzhen, Fundamental Research Funds for the Central Universities of China, No.~30920021127, and the National Natural Science Foundation of China under Grant 62071236.

Yuan-Hsun Lo is with the Department of Applied Mathematics,
National Pingtung University, Pingtung City 90003, Taiwan (e-mail:
\mbox{yhlo@mail.nptu.edu.tw}). 

Kenneth W. Shum is with the School of Science and Engineering, The Chinese University of Hong Kong, Shenzhen 518100, China (email: \mbox{wkshum@cuhk.edu.cn}) (corresponding author).

Wing Shing Wong is with the Department of Information Engineering,
The Chinese University of Hong Kong, Hong Kong (e-mail:
\mbox{wswong@ie.cuhk.edu.hk}).

Yijin Zhang is with the School of Electronic and Optical Engineering,
Nanjing University of Science and Technology, Nanjing 210094, China (e-mail: \mbox{yijin.zhang@gmail.com}).
}}
\fi

\date{}

\ifdefined\SingleColumn
\fi

\newtheorem{theorem}{Theorem}

\newtheorem{definition}{Definition}

\begin{document}

\maketitle

{\bf Abstract:}
Conflict-avoiding codes (CACs) were introduced by Levenshtein as a single-channel transmission scheme for a multiple-access collision channel without feedback.
When the number of simultaneously active source nodes is less than or equal to the weight of a CAC, it is able to provide a hard guarantee that each active source node transmits at least one packet successfully within  a fixed time duration, no matter what the relative time offsets between the source nodes are.
In this paper, we extend CACs to multichannel CACs for providing such a hard guarantee over multiple orthogonal channels.
Upper bounds on the number of codewords for multichannel CACs of weights three and four are derived, and constructions that are optimal with respect to these bounds are presented.

\medskip

{\bf Keywords:} {\em Conflict-avoiding codes, two-dimensional optical orthogonal codes, frequency-hopping sequences, collision channel without feedback, grant-free transmissions.}

\section{Introduction}

A model for multiple access without feedback has been considered previously in~\cite{TL83, MasseyMathys85}, with a focus on the single-channel case.
We extend this model to the multichannel case here.
Consider a wireless network where multiple source nodes want to send their packets to a common sink node through $M \geq 1$ orthogonal channels.
All the channels admit the same time-slotted structure with the same time slot duration.
We assume that the system is {\em slot-synchronous} but not {\em frame-synchronous}, that is, the relative time offsets between the source nodes can be any arbitrary values that are integer multiples of a slot duration.
All packets last for a slot duration, and the slot-synchronous assumption allows them to be transmitted within the slot boundaries.
If two or more packets are transmitted in a time slot on the same channel, we assume that there is a collision and no information can be decoded from the transmitted packets, otherwise all packets are received without any error.
When a source node has some data to be sent, the source node becomes active and transmits packets according to a predefined scheme without feedback from the sink node.

If we adopt ALOHA as a transmission scheme in this model, each active source node is required to transmit a packet in a time slot on a channel with some probability, independent of the other active source nodes~\cite{Abramson73}.
This decentralized scheme has the advantage that it can be implemented without any centralized controller.
However, there is no hard guarantee that each active source node is able to transmit at least one packet successfully within a fixed time duration, due to the probabilistic nature of ALOHA and the lack of feedback.
This is undesirable in mission- and time-critical applications.

To provide such a hard guarantee for the single-channel case, i.e., $M=1$, conflict-avoiding codes (CACs) were introduced by Levenshtein in~\cite{Levenshtein05} as a deterministic transmission scheme.
A CAC is a collection of codewords represented by zero-one sequences, and each source node is assigned a unique codeword from this collection.
When a source node changes its status from inactive to active, it reads out an entry from the assigned codeword sequentially and periodically, and transmits a packet if and only if the entry is equal to~1.
The number of 1's in a codeword, called the {\em Hamming weight} (or simply the \emph{weight}) of the codeword. 
It is required that any pair of distinct codewords in a CAC have at most one overlapping ``1'' regardless of the relative time offsets.
As a result, if all codewords in a CAC have weight $w$ and length $L$, then we can ensure that each active source node can successfully send a packet within any consecutive $L$ slots, provided that the number of simultaneously active source nodes is less than or equal to $w$.
Obviously, the codeword length $L$ measures the worst-case delay that an active source node has to wait until it can send a packet successfully, and the weight $w$ is the maximal number of simultaneously active source nodes that can be supported.

In the literature of CACs, a common design goal is to maximize the total number of source nodes that can be supported, i.e., the cardinality of a set of codewords satisfying the aforementioned requirements, when $L$ and $w$ are given. CACs of weight 3 are studied in~\cite{Levenshtein07, Jimbo, Mishima09, FLM10,Momihara07a, HLS10, MZS14, WF13,MK17}, and CACs of weights four to seven are studied in \cite{LFL15,LMJ16,BT17a, BT17b}.
Results on CACs with general weights are presented in~\cite{Momihara07, SW10, SZW10}.
CACs also find applications in correcting limited-magnitude error in flash memory, and a CAC is usually called a splitter set in this context~\cite{Cassuto,Klove}.

Recently, to support ultra-reliable grant-free transmissions in 5G new radio, the problem formulation of CACs was extended to the multichannel case, i.e., $M>1$, by Chang {\em et~al.}~\cite{CLW19}, with the aim of further increasing the number of source nodes that can be supported.
A \emph{multichannel CAC} (MC-CAC) is a collection of two-dimensional (2-D) codewords represented by zero-one arrays, and any pair of distinct codewords have at most one overlapping ``1'' regardless of the relative time offsets.
As a result, an MC-CAC with $M\times L$ codewords and weight $w$ can support $w$ simultaneously active source nodes over $M$ channels, each of which has a successful transmission within every consecutive $L$ slots.
Similar to the design goal of CACs, this paper aims to maximize the total number of source nodes that can be supported in an MC-CAC, when $M$, $L$ and $w$ are given.

MC-CACs are connected to 2-D optical orthogonal codes (OOC), which find applications in optical code-division multiple-access (CDMA) systems (see e.g. \cite{Shivaleela05, Kwong05}). In 2-D OOC, optical pulses are spread in both time and frequency domains, so that Hamming auto-correlation and Hamming cross-correlation are both less than some pre-defined thresholds. An MC-CAC can be regarded as a 2-D OOC with the threshold for Hamming auto-correlation set to infinity. For MC-CAC with weight 3, related works can be found in \cite{Feng2017, Feng2019}, which investigate 2-D OOC with weight 3 and heterogeneous correlation requirements, namely,  Hamming cross-correlation $\lambda_c\leq 1$ and Hamming auto-correlation $\lambda_a\leq 2$. Another related class of 2-D array codes is optical orthogonal signature pattern codes (OOSPC)~\cite{Kitayama, YangKwong}. In contrast to 2-D OOC and MC-CAC, Hamming auto- and cross-correlations are computed with cyclic shifts on both the time and frequency axes. The nature of the problem is quite different from that of MC-CAC.

In addition, we discuss two practical issues in the design of MC-CACs below.
(i) In the system model considered by Chang {\em et~al.}~\cite{CLW19}, a source node \emph{cannot} transmit two packets in two different channels simultaneously, because it would require multiple wireless transmission units.
In this paper, we shall relax this requirement to consider that a source node may send multiple packets in all channels simultaneously.
An upper bound on the code size of MC-CACs under this setting is certainly an upper bound under the setting in~\cite{CLW19}.
(ii) This paper assumes that the system is slot-synchronous, but the relative time offsets in practice may be arbitrary real numbers~\cite{Zhang11}.
  We can remedy this setting by requiring each active source node to transmit only in the first half of a time slot.
  This strategy sacrifices half of the transmission time, but CACs and MC-CACs will then be applicable.

The remainder of this paper is organized as follows.
In Section~\ref{sec:definition}, the problem formulation and a combinatorial characterization of MC-CACs with general weights is given.
Upper bounds on the code sizes of MC-CACs of weights 3 and 4 are derived in Sections \ref{sec:bound3} and \ref{sec:bound4}, respectively.
Constructions that are optimal with respect to these bounds are presented in Section~\ref{sec:construction3}.
Finally, we close this paper with some concluding remarks.

\section{Formal Definition of Multichannel Conflict-Avoiding Codes}
\label{sec:definition}
The transmission pattern is represented by an $M\times L$ zero-one array, where $M$ denotes the number of channels and $L$ refers to the transmission pattern length.
The rows and columns of the array are further indexed by $\mathcal{M}:=\{0,1,2,\ldots, M-1\}$ and $\mathbb{Z}_L:=\{0,1,2,\ldots, L-1\}$, respectively.
The index set $\mathbb{Z}_L$ is considered as a cyclic group equipped with mod-$L$ addition.
In the case when the time offset is $\tau$, a packet is sent in the $i$-th channel at time slot $t\equiv j\oplus_L \tau$ if and only if the $(i,j)$-entry of the array is equal to 1.

\begin{definition}
\rm
Given positive integers $M$, $L$ and $w$, a set of zero-one $M\times L$ arrays $X_k$, for $k=1,2,\ldots, K$, is called an MC-CAC$(M,L,w)$ if

\ifdefined\SingleColumn
\renewcommand{\theenumi}{(\arabic{enumi})}
\else
\renewcommand{\theenumi}{(\arabic{enumi}}
\fi
\begin{enumerate}
\item
$$ \sum_{i=0}^{M-1} \sum_{j=0}^{L-1} X_k(i,j)=w$$
for all $k=1,2,\ldots, K$, and
\item
$$
  \sum_{i=0}^{M-1} \sum_{j=0}^{L-1} X_k(i,j) X_\ell(i,j \oplus_L \tau ) \leq 1
$$
for $k\neq \ell$ and for all $\tau =0,1,2,\ldots, L-1$, with addition $\oplus_L$ performed in $\mathbb{Z}_L$.
\end{enumerate}
Such an array is called a {\em codeword} in this MC-CAC.
\end{definition}


The first condition requires that the weight of each codeword is equal to $w$.
The second condition is a mathematical formulation of the cross-correlation requirement.
If each source node is assigned a distinct codeword from an MC-CAC$(M,L,w)$, there are at most one collision between two distinct source nodes in every consecutive $L$ time slots, regardless of the relative time offsets.

\medskip

\noindent {\bf Example 1.} The arrays in Fig.~\ref{fig:353} are the codewords of an MC-CAC$(3,5,3)$ of size~8.

\begin{figure*}[t]
\begin{gather*}
\begin{array}{|c|c|c|c|c|} \hline
1&1&1&0&0 \\ \hline
0&0&0&0&0 \\ \hline
0&0&0&0&0  \\ \hline
\end{array}\ \ \
\begin{array}{|c|c|c|c|c|} \hline
0&0&0&0&0 \\ \hline
1&1&1&0&0 \\ \hline
0&0&0&0&0 \\ \hline
\end{array}\ \ \
\begin{array}{|c|c|c|c|c|} \hline
0&0&0&0&0 \\ \hline
0&0&0&0&0 \\ \hline
1&1&1&0&0 \\ \hline
\end{array}\ \ \
\begin{array}{|c|c|c|c|c|} \hline
1&0&0&0&0 \\ \hline
1&0&0&0&0 \\ \hline
1&0&0&0&0 \\ \hline
\end{array}
 \\
\begin{array}{|c|c|c|c|c|} \hline
1&0&0&0&0 \\ \hline
0&1&0&0&0 \\ \hline
0&0&1&0&0  \\ \hline
\end{array}\ \ \
\begin{array}{|c|c|c|c|c|} \hline
1&0&0&0&0 \\ \hline
0&0&1&0&0 \\ \hline
0&0&0&0&1 \\ \hline
\end{array}\ \ \
\begin{array}{|c|c|c|c|c|} \hline
1&0&0&0&0 \\ \hline
0&0&0&1&0 \\ \hline
0&1&0&0&0 \\ \hline
\end{array}\ \ \
\begin{array}{|c|c|c|c|c|} \hline
1&0&0&0&0 \\ \hline
0&0&0&0&1 \\ \hline
0&0&0&1&0 \\ \hline
\end{array}
\end{gather*}
\caption{Codewords of an MC-CAC$(3,5,3)$.}
\label{fig:353}
\end{figure*}

\begin{definition}
\rm
Given the number of channels $M$, the codeword length $L$ and weight $w$, we let
$A(M, L, w)$
denote the largest cardinality among all MC-CAC$(M,L,w)$s, i.e.,
$$
A(M,L,w) := \max \{ |\mathcal{C}|:\, \mathcal{C} \text{ is an MC-CAC$(M,L,w)$} \}.
$$
\end{definition}

{\em Remark:} In the definition of MC-CAC$(M,L,w)$, the $w$ 1's in a codeword can be located anywhere in the $M\times L$ array. If a column of an array contains two or more 1's, then the source node to which the array is assigned needs to transmit two or more packets in a time slot. In order to support multiple simultaneous packet transmissions, a source node must be equipped with multiple transmitters, which is costly to implement in practice. If we impose an additional assumption that each of the column sums in every array is at most~1, then the problem formulation is the same as considered in~\cite{CLW19}.

We next give a combinatorial formulation of MC-CACs.

\begin{definition}
\rm
A {\em scheduling pattern} is a subset of the cartesian product $\mathcal{M}\times \mathbb{Z}_L$ of cardinality~$w$,
$$
S = \{(m_1,t_1), (m_2,t_2),(m_3,t_3), \ldots, (m_w, t_w)\},
$$
with $m_j \in \mathcal{M}$ and $t_j\in \mathbb{Z}_L$ for $j=1,2,\ldots, w$. Given a scheduling pattern $S$, and two indices $i_1$, $i_2\in \mathcal{M}$, define the $(i_1,i_2)$ {\em set of differences} by
$$
D_S(i_1,i_2) := \{t_1 - t_2 \bmod L:\, (i_1,t_1), (i_2,t_2) \in S \}
$$
when $i_1\neq i_2$, and
$$
D_S(i_1,i_2) := \{t_1 - t_2 \bmod L:\, (i_1,t_1), (i_2,t_2) \in S \} \setminus\{0\}
$$
when $i_1=i_2$.
\end{definition}

When $i_1=i_2=i$, we exclude the zero element of $\mathbb{Z}_L$ as an element of $D_S(i,i)$. We note that if $t_1-t_2$ is in $D_S(i_1,i_2)$, then  $t_2-t_1$ is in $D_S(i_2,i_1)$. In particular, when $i_1=i_2$, the set $D_S(i_1,i_2)$ is closed under multiplication by $-1$.

The pairs $(m_1,t_1),(m_2,t_2) \ldots, (m_w,t_w)$ in a scheduling pattern are associated to the positions of 1's in a codeword of an MC-CAC.
Because any cyclic shift of a codeword in the time component is regarded as the same codeword, two scheduling patterns
$$
\{(m_1,t_1), (m_2,t_2),(m_3,t_3), \ldots, (m_w, t_w)\}
$$
and
$$\{(m_1,t_1+\tau), (m_2,t_2+\tau),(m_3,t_3+\tau), \ldots, (m_w, t_w+\tau)\}
$$
are associated with the same codeword, and thus are said to be equivalent to each other. With this equivalence, there is a one-to-one correspondence between a scheduling pattern and a zero-one $M\times L$ array with weight~$w$.

In terms of scheduling pattern and set of differences, the defining requirements of an MC-CAC$(M,L,w)$ can be reformulated as follows.

\begin{theorem}
The $M\times L$ arrays corresponding to a set of $K$ scheduling patterns $\mathcal{S} = \{S_1,S_2,\ldots, S_K\}$ form an MC-CAC$(M,L,w)$ if and only if for any $i_1,i_2\in \mathcal{M}$, the $(i_1,i_2)$ sets of differences $D_{S_j}(i_1,i_2)$, $j=1,2,\ldots,K$, are mutually disjoint.
\end{theorem}

It is notationally convenient to represent the sets of differences by an array.

\begin{definition}
\label{def:array} \rm
Given a scheduling pattern $S$, we define an $M\times M$ array $D_S$, whose rows and columns are indexed by the channels indices in $\mathcal{M}$.
For $i_1, i_2\in \mathcal{M}$, the $(i_1,i_2)$ entry of $D_S$ is a subset of $\mathbb{Z}_L$ and is defined by $D_S(i_1,i_2)$. We will call $D_S$ the {\em array of differences} of the scheduling pattern~$S$.
We define the intersection $D_{S} \cap D_{S'}$ of two $M\times M$ arrays by entry-wise intersection.
\end{definition}

Using the convention in Definition~\ref{def:array}, a collection of $K$ scheduling patterns $\mathcal{S} = \{S_1, S_2,\ldots, S_K\} \subset 2^{\mathcal{M}\times \mathbb{Z}_L}$ corresponds to an MC-CAC$(M,L,w)$ if and only if
\begin{equation}|S|=w \quad  \text{ for all } S\in\mathcal{S},
\label{eq:constant_weight}
\end{equation}
and
\begin{equation}
D_{S} \cap D_{S'} = \emptyset \quad \text{for } S \neq S' \text{ in }\mathcal{S}.
\label{eq:disjoint}
\end{equation}
 (The symbol $\emptyset$ above stands for an $M\times M$ array whose entries are all equal to the empty set.)  Because the arrays of differences  must be entry-wise mutually disjoint, the design of MC-CAC can be regarded as a packing problem.

\medskip

\noindent {\bf Example 1 (cont'd).} The scheduling patterns associated with the arrays in Example 1 are
\begin{align*}
S_1 &= \{ (0,0), (0,1), (0,2) \}, \\
S_2 &= \{ (1,0), (1,1), (1,2) \}, \\
S_3 &= \{ (2,0), (2,1), (2,2) \}, \\
S_4 &= \{ (0,0), (1,0), (2,0) \}, \\
S_5 &= \{ (0,0), (1,1), (2,2) \}, \\
S_6 &= \{ (0,0), (1,2), (2,4) \}, \\
S_7 &= \{ (0,0), (1,3), (2,1) \}, \\
S_8 &= \{ (0,0), (1,4), (2,3) \}.
\end{align*}
Each of the above sets is considered as a subset of the cartesian product $\{0,1,2\}\times \mathbb{Z}_5$.
The associated arrays of differences are shown in Fig.~\ref{fig:difference_matrices}.

\begin{figure*}[t]
\begin{align*}
D_{s_1} &= \begin{array}{|c|c|c|} \hline
\{1,2,3,4\} & \phantom{\{1,2,3,4\} }&  \phantom{\{1,2,3,4\} }\\ \hline
&& \\ \hline
&& \\ \hline
\end{array}
\qquad \qquad
D_{s_4} = \begin{array}{|c|c|c|} \hline
&\{0\} &\{0\} \\ \hline
\{0\} &&\{0\} \\ \hline
\{0\} &\{0\}& \\ \hline
\end{array}
\qquad \qquad 
D_{s_7} = \begin{array}{|c|c|c|} \hline
&\{2\} &\{4\} \\ \hline
\{3\} &&\{2\} \\ \hline
\{1\} &\{3\}& \\ \hline
\end{array}
 \\ 
D_{s_2} &= \begin{array}{|c|c|c|} \hline
&& \\ \hline
 \phantom{\{1,2,3,4\} }&  \{1,2,3,4\} & \phantom{\{1,2,3,4\} }\\ \hline
&& \\ \hline
\end{array}
\qquad \qquad
D_{s_5} = \begin{array}{|c|c|c|} \hline
&\{4\} &\{3\} \\ \hline
\{1\} &&\{4\} \\ \hline
\{2\} &\{1\}& \\ \hline
\end{array}
\qquad \qquad
D_{s_8} = \begin{array}{|c|c|c|} \hline
&\{1\} &\{2\} \\ \hline
\{4\} &&\{1\} \\ \hline
\{3\} &\{4\}& \\ \hline
\end{array}
\\
D_{s_3} &= \begin{array}{|c|c|c|} \hline
&& \\ \hline
&& \\ \hline
 \phantom{\{1,2,3,4\} }&   \phantom{\{1,2,3,4\} } & \{1,2,3,4\} \\ \hline
\end{array}
\qquad \qquad 
D_{s_6} = \begin{array}{|c|c|c|} \hline
&\{3\} &\{1\} \\ \hline
\{2\} &&\{3\} \\ \hline
\{4\} &\{2\}& \\ \hline
\end{array}
\end{align*}
\caption{The arrays of differences of the codewords in Example 1. The empty space means that the corresponding entry is the empty set.}
\label{fig:difference_matrices}
\end{figure*}

\section{An Upper Bound on Code Size for Weight 3}
\label{sec:bound3}
In this section we consider MC-CAC$(M,L,3)$s with the number of channels $M\geq 3$.
Scheduling patterns consisting of three packets can be classified into three types:

\ifdefined\SingleColumn
\renewcommand{\theenumi}{(\roman{enumi})}
\else
\renewcommand{\theenumi}{(\roman{enumi}}
\fi

\begin{enumerate}
\item All three packets are sent in the same channel.
\item The three packets are transmitted in three distinct channels.
\item Two packets are transmitted in a channel, and the third one is transmitted in another channel.
\end{enumerate}

In the followings, we analyze the corresponding sets of differences
for each of the three types.

\smallskip

\noindent \underline{Type (i).} When all three packets are sent in the same channel, say channel $m$, the scheduling pattern $S$ has the form
$$
S = \{(m,t_1), (m,t_2), (m,t_3) \}
$$
where $t_1$, $t_2$, and $t_3$ are three distinct elements in $\mathbb{Z}_L$.  Without loss of generality, we may assume that $t_1=0$, and hence it is sufficient to consider scheduling pattern
$$
S = \{(m,0), (m,a), (m,b) \},
$$
where $a$ and $b$ are nonzero elements in $\mathbb{Z}_L$ and $a\neq b$.
The $(m,m)$-entry in the array of differences is
$$
D_S(m,m) = \{ \pm a, \pm b, \pm (b-a)\}.
$$
The set of differences $D_S(m,m)$ contains at most six distinct differences.
This is the only entry in the array of differences that is nonempty.

A scheduling pattern of type (i) is called {\em equi-difference} if $b=2a$, i.e., if $S$ can be written in the form
\begin{equation}
S = \{(m,0), (m,a), (m,2a) \}.
\label{eq:equi_difference1}
\end{equation}
The set of differences in the equi-difference case is equal to $\{\pm a, \pm 2a\}$. We remark that an equi-difference scheduling pattern can also be written as
\begin{equation}
S = \{(m,-a), (m,0), (m,a) \}.
\label{eq:equi_difference2}
\end{equation}
The scheduling patterns in \eqref{eq:equi_difference1} and \eqref{eq:equi_difference2} are equivalent and they are associated with the same codeword in an MC-CAC.
In the sequel we will use the notation in~\eqref{eq:equi_difference1} for equi-difference scheduling patterns.

We list the sizes of all possible sets of differences for codewords of the first type below.
\begin{align}
&|D_S(m,m)| = \notag \\
&\begin{cases}
2 & \text{if } 3|L \text{ and } S = \{(m,0), (m,L/3) , (m,2L/3)\}, \\
3 & \text{if } 4|L \text{ and } S = \{(m,0), (m,L/4) , (m,L/2)\}, \\
5 & \text{if } 2|L \text{ and } S = \{(m,0), (m,a), (m,L/2)\} \\
  & \qquad \ \text{ for some } a\neq L/4, \\
6 & \text{if $S$ is non-equi-difference and } (m,L/2)\not\in S, \\
4 & \text{otherwise}.
\end{cases} \label{eq:same_channel}
\end{align}
The sets of differences of the two scheduling patterns 
\begin{gather}
\{(m,0), (m,L/3) , (m,2L/3)\}, \label{eq:special1}\\
 \{(m,0), (m,L/4) , (m,L/2)\} \label{eq:special2}
\end{gather}
have size 2 and 3, respectively. An MC-CAC may contain the scheduling pattern in \eqref{eq:special1} only when $L$ is divisible by 3, and the the scheduling pattern in \eqref{eq:special2} only when $L$ is divisible by~4.

We summarize the sizes of sets of differences of a scheduling pattern 
$$S=\{(m,0), (m,a), (m,b)\}$$
 of type (i) as follows,
$$
|D_S(i_1,i_2)| =
 \begin{cases}0 & \text{if }i_1 \neq m \text{ or } i_2 \neq m ,\\
2,3,4,5, \text{or } 6 & \text{if } i_1=i_2 = m.
\end{cases}
$$

\underline{Type (ii).} A scheduling pattern of the second type is of the form
$$
S = \{ (m_1,t_1), (m_2,t_2), (m_3,t_3)\},
$$
where $m_1$, $m_2$ and $m_3$ are distinct channel indices. The sizes of the sets of differences are
$$|D_S(i_1,i_2)| = 
  \begin{cases}
1 & \text{if } i_1 \neq i_2 \text{ and }\{i_1,i_2\} \subset \{m_1,m_2,m_3\} ,\\
0 & \text{otherwise}.
\end{cases}
$$

\underline{Type (iii).} A scheduling pattern of the third type can be written as
$$
S = \{ (m_1,t_1), (m_1,t_2), (m_2, t_3)\},
$$
where $m_1\neq m_2$ and $t_1\neq  t_2$. If we want the set of difference to be nonempty, the only possibilities are $\{i_1,i_2\}=\{m_1,m_2\}$ or $i_1=i_2=m_1$.
The sizes of the sets of differences are
\begin{align}
&|D_S(i_1,i_2)| = \notag \\
&  \begin{cases}
0 & \text{if } \{i_1 , i_2\} \neq \{m_1\} \text{ and } \{i_1,i_2\} \neq \{ m_1, m_2\}, \\
1 & \text{if } \{i_1,i_2\}=\{m_1\} \text{ and } 2|L \text{ and } t_1 - t_2 \equiv L/2 \bmod L, \\ 	
2 & \text{otherwise.}
\end{cases} \label{eq:type_iii}
\end{align}

\medskip

By analyzing the three types of codewords, we can obtain an upper bound on the total number of codewords, which is stated in the following theorem.

\begin{theorem}
For $M \geq 3$, we have the following upper bound on the cardinality of an $(M,L,3)$ MC-CAC,
\begin{align*}
&A(M,L,3) \leq \\
&\begin{cases}
 \big\lfloor \frac{M}{12} (2ML + L + 6) \big\rfloor & \text{if }L = 0,6 \bmod 12, \\
 \big\lfloor \frac{M}{12} (2ML + L +3) \big\rfloor & \text{if }L = \pm 3 \bmod 12, \\
 \big\lfloor \frac{M}{12} (2ML + L ) \big\rfloor & \text{if }L = \pm2, \pm 4 \bmod 12, \\
 \big\lfloor \frac{M}{12} (2ML + L - 3) \big\rfloor & \text{if }L = \pm 1, \pm 5 \bmod 12. 
\end{cases}
\end{align*}
\label{thm:bound1}
\end{theorem}

\begin{proof}
Let $\mathcal{C}$ be an $(M,L,3)$ MC-CAC with $M\geq 3$. 
We first define some variables. For $i=1,2,\ldots, M$, we let $N_i$ be the number of codewords in $\mathcal{C}$ whose scheduling patterns are in the form
$\{(i,t_1), (i,t_2), (i, t_3)\}$,
for some $t_1, t_2, t_3\in \mathbb{Z}_L$. 

For any two distinct channel indices $i$ and $j$, we let $N_{i,j}$ be the number of codewords of type (iii) in $\mathcal{C}$, whose scheduling patterns are in the form
$\{ (i,t_1), (i,t_2), (j,t_3)\}$, for some $t_1, t_2, t_3\in \mathbb{Z}_L$. 

For any three distinct channel indices $i$, $j$ and $k$, we let $N_{\{i,j,k\}}$ be the number of codewords of type (ii) in $\mathcal{C}$ whose scheduling patterns are in the form
$\{ (i,t_1), (j,t_2), (k,t_3) \}$. 
The total number of codewords in $\mathcal{C}$ is thus equal to
\begin{equation}
N = \sum_{i\in\mathcal{M}} N_i + \sum_{\substack{i,j\in\mathcal{M} \\ i\neq j}} N_{i,j} + \sum_{\substack{S \subset \mathcal{M} \\ |S|=3}} N_{S}.
\label{eq:N}
\end{equation}
The second summation in the above line is extended over all ordered pairs $(i,j)$ in $\mathcal{M}\times \mathcal{M}$ with distinct components. In the third summation in~\eqref{eq:N}, the set $S$ runs through all subsets of $\mathcal{M}$ with cardinality~3.

Consider two channel indices $\alpha$, $\beta \in \mathcal{M}$. If $\alpha\neq \beta$, the sets of differences $D_S(\alpha,\beta)$, with $S$ going through all scheduling patterns in $\mathcal{C}$, must be disjoint. Hence, we have
\begin{equation}
 2N_{\alpha,\beta} + 2N_{\beta,\alpha}+ \sum_{\substack{A \supset \{ \alpha, \beta\} \\ |A|=3}} N_{A} \leq |\mathbb{Z}_L| = L.
\label{eq:LP1}
\end{equation}
The summation in \eqref{eq:LP1} is over all subsets $A$ in $\mathcal{M}$ with cardinality~3 that contain both channel indices $\alpha$ and $\beta$.  We note that if we fix a subset $A' \subset \mathcal{M}$ with cardinality~3, and let $\alpha$ and $\beta$ running through all pairs of distinct channel indices, the variable $N_{A'}$ appears precisely 3 times among the inequalities in \eqref{eq:LP1}. By summing \eqref{eq:LP1} over all sets of distinct channel indices $\{\alpha, \beta\}$, we obtain
\begin{equation}
2 \sum_{(i,j)\in\mathcal{M}^2} N_{i,j} + 3 \sum_{A \in \binom{\mathcal{M}}{3}} N_A \leq \binom{M}{2} L.
\label{eq:LP1a}
\end{equation}

Next consider $D_S(\alpha,\alpha)$ for $\alpha\in\mathcal{M}$. Note that $D_S(\alpha,\alpha)$ is non-empty only when $S$ is a codeword of type (i) or (iii). As the sets of differences $D_S(\alpha,\alpha)$ are mutually disjoint over all $S\in\mathcal{C}$, with the union being a subset of $\mathbb{Z}\setminus\{0\}$, we have
\begin{equation}
\sum_{S\in \mathcal{C}} |D_S(\alpha,\alpha)| \leq L-1, \qquad \text{ for } \alpha\in\mathcal{M}.
\label{eq:LP1b}
\end{equation}

Without loss of generality, we put all scheduling patterns in the form $\{(i_1,0), (i_2,a), (i_3,b)\}$. Since, by~\eqref{eq:same_channel}, the equi-difference codeword with generator $L/3$ and $L/4$ has relatively small set of differences, we define the following indicator functions. For $\alpha\in\mathcal{M}$, let 
\begin{align*}
E_\alpha &:= \begin{cases}
1 & \text{if } \{(\alpha,0), (\alpha,L/3), (\alpha,2L/3)\} \in \mathcal{C}, \\
0 & \text{otherwise},
\end{cases} \\
\intertext{and}
F_\alpha &:= \begin{cases}
1 & \text{if } \{(\alpha,0), (\alpha,L/4), (\alpha,L/2)\} \in \mathcal{C}, \\
0 & \text{otherwise}.
\end{cases}
\end{align*}
We further define, for $\alpha\in\mathcal{M}$, that
$$
G_\alpha := 
\begin{cases}
1 & \text{ if } \{(\alpha,0), (\alpha,a), (\alpha,L/2)\}\in\mathcal{C} \\
  & \qquad \text{ for some } a\neq L/4, \\
0& \text{ otherwise}.
\end{cases}
$$
The set of differences associated to the codeword in the definition of $G_\alpha$ has size~5 (see \eqref{eq:same_channel}).

Similarly, we assume all type (iii) scheduling patterns are in the form $S=\{(\alpha,0), (\alpha, a), (\alpha',b)\}$, with $\alpha\neq\alpha'$. From~\eqref{eq:type_iii}, we see that the codeword $S=\{(\alpha,0), (\alpha, L/2), (\alpha',b)\}$ contributes a size of 1 on the set of differences $D_S(\alpha,\alpha)$. Define
$$
H_\alpha := \begin{cases}
1 & \text{ if } \{ (\alpha,0), (\alpha,L/2), (\alpha',b)\} \in\mathcal{C} \\
  & \qquad \text{ for some } \alpha\neq\alpha' , \\
0 & \text{ otherwise}.
\end{cases}
$$
Note that $F_\alpha + G_\alpha+H_\alpha \leq 1$, since each of these kinds of  scheduling patterns contains $L/2$ in their sets of differences $D_S(\alpha,\alpha)$. At most one of them can be included in a multi-channel CAC.

The rest of the proof is divided into three cases according to the residue of $L$ modulo~4.
We will use the short-hand notation $\mathbf{1}_{\{P\}}$, that is defined to be~1 if $P$ is a true statement, and 0 otherwise.

{\em Case 1, $L \equiv 1 \bmod 2$.} In this case, we have $F_\alpha=G_\alpha=H_\alpha=0$ for all channel indices $\alpha$ in $\mathcal{M}$. When $S$ is of type (i), by~\eqref{eq:same_channel}, $|D_S(\alpha,\alpha)|=2$ if $S=\{(\alpha,0), (\alpha,L/3), (\alpha, 2L/3)\}$ and 
$|D_S(\alpha,\alpha)|\geq 4$ otherwise. When $S$ is of type (iii), by~\eqref{eq:type_iii}, $|D_S(\alpha,\alpha)|=2$. By going over all codewords in $\mathcal{C}$, it follows from~\eqref{eq:LP1b} that
$$
4(N_\alpha -E_\alpha) + 2E_\alpha + 2 \sum_{\gamma\in\mathcal{M}\setminus\{\alpha\}} N_{\alpha,\gamma} \leq L-1.
$$
Using the fact that $E_\alpha=1$  only when $L$ is divisible by 3, i.e., $E_\alpha \leq \mathbf{1}_{\{L\in3\mathbb{Z}\}}$, we have
$$
4N_\alpha + 2 \sum_{\gamma\in\mathcal{M}\setminus\{\alpha\}} N_{\alpha,\gamma} \leq L-1 + 2 \mathbf{1}_{\{L\in3\mathbb{Z}\}}.
$$
By going through all channel indices $\alpha$, we obtain
\begin{equation}
4 \sum_{i\in\mathcal{M}} N_i + 2 \sum_{(i,j)\in\mathcal{M}^2} N_{i,j} \leq M(L-1 + 2 \mathbf{1}_{\{L\in3\mathbb{Z}\}}).
\label{eq:LP1c}
\end{equation}

We next multiply the inequality in~\eqref{eq:LP1a} by 4, multiply the inequality in~\eqref{eq:LP1c} by 3, and sum the resulting inequalities. This yields
\begin{align*}
12 \sum_{i\in\mathcal{M}} N_i + &
14 \cdot\sum_{\substack{i,j\in \mathcal{M} \\ i\neq j}} N_{i,j} +
12 \sum_{\substack{A\subseteq \mathcal{M} \\ |A|=3}} N_A \\
& \leq M(2ML+L-3+6\mathbf{1}_{L\in 3\mathbb{Z}}).
\end{align*}
The summation in the middle is multiplied by 14. If we replace 14 by 12, by~\eqref{eq:N}, the left-hand side in the above inequality is the same as $12N$. We then have
\begin{equation}
 N  \leq \frac{M}{12}(2ML +L -3 + 6 \mathbf{1}_{\{L\in 3\mathbb{Z}\}}).
\label{eq:bound_odd}
\end{equation}

\smallskip

{\em Case 2, $L \equiv 2 \bmod 4$.} In this case we have $F_\alpha=0$ since $L/4$ is not an integer. Notice that a codeword can be $\{(\alpha,0), (\alpha,L/3), (\alpha,2L/3)\}$ or one of the following two kinds of codewords: $\{(\alpha,0), (\alpha,a), (\alpha,L/2) \}$ and $\{(\alpha,0), (\alpha,L/2), (\alpha',a)\}$ for some $\alpha'\neq \alpha$. By the definition of $E_\alpha$, $G_\alpha$ and $H_\alpha$, we can conclude from \eqref{eq:same_channel} and \eqref{eq:type_iii} that, when a codeword $S$ is of type (i), we have
$$
|D_S(\alpha,\alpha)| \begin{cases}
=2 & \text{ if } E_\alpha=1,\\
=5 & \text{ if } G_\alpha=1,\\
\geq 4 & \text{ otherwise},
\end{cases}
$$
and when $S$ is of type (iii), we have
$$
|D_S(\alpha,\alpha)| \begin{cases}
=1 & \text{ if } H_\alpha=1,\\
=2 & \text{ otherwise.}
\end{cases}
$$
By summing the inequality in~\eqref{eq:LP1b} over all codewords in $\mathcal{C}$, we obtain
\begin{align*}
&4  (N_\alpha-E_\alpha-G_\alpha) + 2E_\alpha + 5G_\alpha \\
&\qquad  + 2 \Big( \sum_{\gamma\in\mathcal{M}\setminus\{\alpha\}} N_{\alpha,\gamma} - H_\alpha \Big)+H_\alpha \leq  L-1.
\end{align*}
Using the fact that $E_\alpha \leq \mathbf{1}_{\{L\in3\mathbb{Z}\}}$, $G_\alpha \geq 0$ and $H_\alpha \leq 1$, we can simplify the above inequality to
$$
4N_\alpha + 2 \sum_{\gamma\in\mathcal{M}\setminus\{\alpha\}} N_{\alpha,\gamma}
\leq L + 2 \mathbf{1}_{\{L\in3\mathbb{Z}\}}.
$$
The sum of the above inequality over all channel indices $\alpha\in\mathcal{M}$ is
\begin{equation}
4\sum_{i\in\mathcal{M}} N_i + 2 \sum_{(i,j)\in\mathcal{M}^2} N_{i,j}
\leq M(L + 2 \mathbf{1}_{\{L\in3\mathbb{Z}\}}).
\label{eq:LP1d}
\end{equation}

By the same argument as in Case 1, we add 4 times the inequalities in \eqref{eq:LP1a} and 3 times the inequality in \eqref{eq:LP1b} to get
\begin{align*}
& \phantom{=} 12 \sum_{i\in\mathcal{M}} N_i + 14 \sum_{(i,j)\in\mathcal{M}^2} N_{i,j}
+ 12 \sum_{A\in\binom{\mathcal{M}}{3}} N_A \\
&\leq M(2ML + L + 6\mathbf{1}_{\{L\in3\mathbb{Z}\}}).
\end{align*}
After decreasing the multiplicative factor 14 on the left-hand side to 12, we can simplify the inequality to
\begin{align}
 N \leq \frac{M}{12}(2ML +L + 6\mathbf{1}_{\{L\in 3\mathbb{Z}\}} ) .
\label{eq:bound_even1}
\end{align}

\smallskip

{\em Case 3, $L \equiv 0 \bmod 4$.} In this case, $\mathcal{C}$ may contain the codewords whose scheduling patterns listed in the conditions of the definitions of $E_\alpha$, $F_\alpha$, $G_\alpha$ and $H_\alpha$. By a similar argument as in Case 2, we have
\begin{align*}
&4(N_\alpha - E_\alpha - F_\alpha - G_\alpha) + 2E_\alpha + 3F_\alpha+5G_\alpha \\
& \qquad \qquad \quad + 2 \Big( \sum_{\gamma\in\mathcal{M}\setminus\{\alpha\}} N_{\alpha,\gamma} - H_\alpha \Big) + H_\alpha \leq L-1
\end{align*}
for each $\alpha\in\mathcal{M}$. Again, we use the fact that $E_\alpha \leq \mathbf{1}_{\{L\in 3\mathbb{Z}\}}$,  $G_\alpha \geq 0$ and $F_\alpha+H_\alpha\leq 1$, we can simplify the previous inequality to
$$
4N_\alpha +  2\sum_{\gamma\in\mathcal{M}\setminus\{\alpha\}} N_{\alpha,\gamma} 
\leq L + 2 \mathbf{1}_{\{L\in 3\mathbb{Z}\}}.
$$
Sum the above inequality over all $\alpha\in\mathcal{M}$, we obtain an inequality that is identical to~\eqref{eq:LP1d}. We can continue the same argument as in Case 2 and obtain~\eqref{eq:bound_even1}.

Combining \eqref{eq:bound_odd} and~\eqref{eq:bound_even1}, we obtain
$$
 N \leq \frac{M}{12}(2ML +L -3+3\mathbf{1}_{\{L\in 2\mathbb{Z}\}} +6\mathbf{1}_{\{L\in 3\mathbb{Z}\}} ) .
$$
This completes the proof of Theorem~\ref{thm:bound1}.
\end{proof}

Using the bound in the Theorem~\ref{thm:bound1}, we can show that the MC-CAC in Example~1 is optimal.
More precisely, plugging $M=3$ and $L=5$ into Theorem~\ref{thm:bound1}, we obtain
\begin{align*}
A(3,5,3) &\leq \frac{3}{12} (2\cdot 3 \cdot 5 +  5 - 3) = 32/4 = 8.
\end{align*}
This proves that the MC-CAC$(5,3,5)$  in Example~1 is optimal.

\section{An Upper Bound on Code Size for Weight 4}
\label{sec:bound4}

The main theorem in this section provides an upper bound of the size of an MC-CAC$(M,L,4)$ with $M \geq 4$.
The proof technique is similar to that in the previous section.


\begin{theorem} For $M\geq 4$ and positive integer $L$,
\begin{align*}
&A(M,L,4) \leq \\
& \begin{cases}
\big\lfloor \frac{M}{12}(ML+L)\big\rfloor+1  & \text{if } L =0 \bmod 60, \\
\big\lfloor \frac{M}{12}(ML+L+8) \big\rfloor & \text{if } L =\pm12, \pm20,\pm24, 30 \bmod 60, \\
\big\lfloor \frac{M}{12}(ML+L+6) \big\rfloor & \text{if } L =\pm15 \bmod 60, \\
\big\lfloor \frac{M}{12}(ML+L+4) \big\rfloor & \text{if } L =\pm4, \pm6, \pm8 ,\pm10, \\
&\phantom{\text{if } L=} \pm16, \pm18, \pm28 \bmod 60, \\
\big\lfloor \frac{M}{12}(ML+L+2) \big\rfloor & \text{if } L =\pm3, \pm 5, \pm9, \pm21,\\
&\phantom{\text{if } L=} \pm21, \pm25, \pm27 \bmod 60, \\
\big\lfloor \frac{M}{12}(ML+L) \big\rfloor & \text{if } L =\pm2, \pm14, \pm 22,  \pm26  \bmod 60 ,\\
\big\lfloor \frac{M}{12}(ML+L-2) \big\rfloor & \text{if } gcd(L,60)=1.
\end{cases}
\end{align*}
\label{thm:b4}
\end{theorem}

In order to derive the bound in Theorem~\ref{thm:b4}, we classify the codewords in an MC-CAC$(M,L,4)$  into five types.

\ifdefined\SingleColumn
\renewcommand{\theenumi}{(\roman{enumi})}
\else
\renewcommand{\theenumi}{(\roman{enumi}}
\fi

\begin{enumerate}
\item All four packets are sent in the same channel.
\item The four packets are transmitted in four distinct channels.
\item One packet is transmitted in a channel, and the other three packets are transmitted in another channel.
\item Two packets are transmitted in a channel, and the other two packets are transmitted in another channel.
\item The packets are transmitted in three channels.
\end{enumerate}

We first analyze the set of differences of codewords of each of the five types before we give the proof of Theorem~\ref{thm:b4}.

\medskip

\noindent \underline{Type (i).} The scheduling pattern $S$ when all four packets are sent in the same channel can be put in the form
$$
S = \{ (m,0), (m, a), (m, b), (m, c)\},
$$
where $a, b, c$ are nonzero and distinct elements in $\mathbb{Z}_L$. We may assume $0<a<b<c$ without loss of generality.
We call this scheduling pattern {\em equi-difference} if 0, $a$, $b$ and $c$ form an arithmetic progression. In this case, we may assume $b=2a$ and $c=3a$, and the set of nonzero differences
$$
D_S(m,m) = \{\pm a, \pm 2a, \pm 3a\}
$$
has size at most~6.

The size of $D_S(m,m)$ can be strictly less than~6. This happens when
the scheduling pattern $S$ is one of the following scheduling patterns:

\smallskip

\noindent (a) $S = \{0, L/4, L/2, 3L/4\}$ when $L$ is divisible by 4;

\smallskip

\noindent (b) $S$ is any subset of $\{0, L/5, 2L/5, 3L/5, 4L/5 \}$ with size 4 when $L$ is divisible by 5;

\smallskip

\noindent (c) $S = \{0, a, L/2, L/2+a\}$ or $S = \{0, a, L/2, L-a\}$ when $L$ is divisible by 2 and $0< a <L/2$.

The sizes of the sets of differences $D_S(m,m)$ in parts (a), (b) and (c) are 3, 4 and 5, respectively. We note that an MC-CAC cannot contain scheduling patterns from both (a) and (c), because the corresponding sets of differences contain $L/2$ as a common element.
In summary, when the scheduling pattern $S$ can be written as in (a), (b) and (c) above, then $|D_S(m,m)|$ is strictly less than $6$; 
otherwise, the size of $D_S(m,m)$ is larger than or equal to~6.

\smallskip

\noindent \underline{Type (ii).} When all four packets are on distinct channels, a scheduling pattern can be written as
$$
S = \{(m_1, t_1), (m_2, t_2), (m_3, t_3), (m_4, t_4) \},
$$
where $m_1$, $m_2$, $m_3$ and $m_4$ are distinct channel indices in $\mathcal{M}$. The size of the sets of differences is equal to
$$
|D_S(i_1,i_2)| = \begin{cases}
1 & \text{ if } i_1 \neq i_2 \text{ and }\{i_1,i_2\} \subset \{m_1,\ldots, m_4\}, \\
0 & \text{ otherwise.}
\end{cases}
$$

\smallskip
\noindent \underline{Type (iii).} Consider a scheduling pattern
$$
S = \{ (m_1,0), (m_2,a), (m_2,b), (m_2,c)\},
$$
where $m_1$ and $m_2$ are two distinct channel indices in $\mathcal{M}$, and $a,b,c$ are three distinct time indices in $\mathbb{Z}_L$.
The analysis of the difference structure of the three packets $(m_2,a)$, $(m_2,b)$ and $(m_2,c)$ in the same channel can be done as in the previous section.
We let $S' = \{(m_2,a), (m_2,b), (m_2,c)\}$.
The size of $D_S(i,j)$ can be determined as
$$
|D_S(i,j)| = \begin{cases}
3 & \text{ if } \{i, j\}=\{m_1,m_2\}, \\
|D_{S'}(i,j)| & \text{ if } i=j=m_2 , \\
0 & \text{ otherwise}.
\end{cases}
$$
From the analysis in the previous section, we know that when $i=j=m_2$, the smallest value of $|D_{S'}(i,j)|$ can be equal to 2, and it  happens only when $L$ is divisible by $3$.
The second smallest possible value of $|D_{S'}(i,j)|$ is 3, and it occurs only when $L$ is divisible by $4$.
When $L$ is neither divisible by 3 nor divisible by 4, the size of $D_{S'}(i,j)$ is larger than or equal to~4.

\smallskip
\noindent \underline{Type (iv).}  A scheduling pattern of type (iv) can be written as
$$
S = \{ (m_1,a), (m_1,b), (m_2,c), (m_2,d)\},
$$
where $m_1$ and $m_2$ are two distinct channel indices, $a$ and $b$ are distinct time indices, and $c$ and $d$ are another two distinct time indices. The size of $D_S(i,j)$ is
$$
|D_S(i,j)| = \begin{cases}
2 & \text{ if } i=j\in\{m_1,m_2\}, \\
2,3,4 & \text{ if } \{i,j\} = \{m_1,m_2\}, \\
0 & \text{ otherwise}.
\end{cases}
$$
When $i\neq j$, the size of $D_S(i,j)$ is equal to 2 when $L$ is even, $a=c$, $b=d$ and $b-a=L/2 \bmod L$. If $a=c$, $b=d$ but $b-a$ is not congruent to $L/2$, then $|D_S(i,j)|=3$.

\smallskip
\underline{Type (v).} Consider a scheduling pattern
$$
S = \{(m_1,a), (m_1,b), (m_2,c), (m_3,d) \},
$$
where $m_1$, $m_2$ and $m_3$ are distinct channel indices, and $c\neq d$. We can determine the size of the set of differences as
\begin{align*}
& |D_S(i,j)| = \\
& \begin{cases}
1 & \text{ if } \{i,j\} = \{m_2,m_3\} , \\
2 & \text{ if } \{i,j\} = \{m_1,m_2\} \text{ or } \{i,j\} = \{m_1,m_3\}, \\
1,2 & \text{ if } i=j=m_1, \\
0 & \text{ otherwise}.
\end{cases}
\end{align*}
We note that the case $|D_S(m_1,m_1)|=1$ happens only when $L$ is even.

\begin{proof}[Proof of Theorem~\ref{thm:b4}]
We adopt the following notation. Given a finite set $A$, we denote the collection of all subsets of $A$ with size $k$ by $\binom{A}{k}$.

Let $\mathcal{C}$ be an MC-CAC$(M,L,4)$ with $M\geq 4$.
For $i\in \mathcal{M}$, let $N_i$ be the number of codewords of type (i) in $\mathcal{C}$.
For $B \in \binom{ \mathcal{M}}{4}$, we let $N_B$ denote the number of codewords of type (ii) in which the packets are sent in the four channels with indices in set~$B$.
For two distinct channel indices $i$ and $j$ in $\mathcal{M}$, we let $N_{i,j}$ be the number of codewords of type (iii) in which 1 packet is in channel $i$ and the other 3 packets are in channel~$j$.
For any subset $A=\{i,j\}$ of $\mathcal{M}$ with size 2, let $N_A$ be the number of codewords of type (iv) in which 2 packets are in channel $i$ and the other 2 packets are in channel~$j$.
Finally, for channel index $i$ in $\mathcal{M}$, and a subset $A=\{j,k\}\subset \mathcal{M}\setminus\{ i\}$ with size 2, we let $N_{i,A}$ be the number of codewords in $\mathcal{C}$ such that 2 packets are sent in channel~$i$, one packet is in channel $j$ and one is in channel~$k$. We can thus decompose the total number of codewords in $\mathcal{C}$ into five parts
\begin{align}
N&= \sum_{i=1}^M N_i + \sum_{B \in \binom{\mathcal{M}}{4}} N_B + \sum_{i\neq j} N_{i,j} + \sum_{A\in\binom{\mathcal{M}}{2}} N_A \notag \\
& \phantom{=}
+\sum_{i=1}^M \sum_{A\in \binom{\mathcal{M}\setminus\{i\}}{2}} N_{i,A}.
\label{eq:N4}
\end{align}

For each $\alpha\in \mathcal{M}$, because $D_S(\alpha,\alpha)$ must be mutually disjoint when $S$ ranges over all codewords in $\mathcal{C}$, we have the following inequality
\begin{align}
&6N_\alpha + 4 \sum_{\ell \neq \alpha} N_{\ell,\alpha} +
2 \sum_{\ell \neq \alpha} N_{\{\alpha,\ell\}} + 2 \sum_{A \in \binom{\mathcal{M}\setminus\{\alpha\}}{2}} N_{\alpha,A} \notag \\
&   \leq L-1 + \mathbf{1}_{\{L\in 2\mathbb{Z}\}} + 2\mathbf{1}_{\{L\in3\mathbb{Z}\}} + 2\mathbf{1}_{\{L\in4\mathbb{Z}\}} + 2\mathbf{1}_{\{L \in5\mathbb{Z}\}}
\label{eq:C}
\end{align}
for each $\alpha\in\mathcal{M}$. There are $M$ such inequalities.

Let $\alpha$ and $\beta$ be two channel indices in $\mathcal{M}$ with $\alpha<\beta$. Because the set of differences $D_S(\alpha,\beta)$ must be mutually disjoint when we consider all scheduling patterns $S$ in $\mathcal{C}$, we have
\begin{align}
& \quad \sum_{\substack{B \in \binom{\mathcal{M}}{4} \\ \{\alpha,\beta\} \subset B}} N_B
+ 3 (N_{\alpha,\beta} + N_{\beta,\alpha}+ N_{\{\alpha,\beta\}}) \notag \\
& +  \sum_{\ell \in \mathcal{M}\setminus\{\alpha,\beta\}}( N_{\ell, \{\alpha,\beta\}}   + 2 N_{\alpha,\{\beta,\ell\}} + 2N_{\beta, \{\alpha,\ell\}})
  \leq L.
\label{eq:D}
\end{align}
Note that in \eqref{eq:D} we only consider $\alpha<\beta$. We will get the same equation if we swap the role of $\alpha$ and $\beta$. As a result, we have $M(M-1)/2$ inequalities in the form of~\eqref{eq:D}.

We sum over the $M+M(M-1)/2$ inequalities to get
\begin{align*}
&6 \sum_{i\in\mathcal{M}} N_i + 6 \sum_{B \in \binom{\mathcal{M}}{4}} N_B
+ 7 \sum_{i\neq j} N_{i,j}
+ 7 \sum_{A \in \binom{\mathcal{M}}{2}} N_{A} \\
& \!+ 7 \sum_{A \in \binom{\mathcal{M}}{2}} \sum_{\ell \in \mathcal{M}\setminus A} N_{\ell,A} \leq M(L-1)  + \frac{M(M-1)L}{2} +MJ,
\end{align*}
where $J$ is defined as
\begin{equation}
J =  \mathbf{1}_{\{L\in 2\mathbb{Z}\}} + 2\mathbf{1}_{\{L\in3\mathbb{Z}\}} + 2\mathbf{1}_{\{L\in4\mathbb{Z}\}} + 2\mathbf{1}_{\{L \in5\mathbb{Z}\}}.
\label{eq:J}
\end{equation}
We note that the coefficients on the left-hand side is either 6 or 7. By replacing 7 by 6, and combining with the expression of $N$ in~\eqref{eq:N4}, we obtain
\begin{align*}
N &\leq \frac{1}{12} [2M(L-1)  + M(M-1)L +2MJ] \\
&= \frac{M}{12} [2(L-1)  + (M-1)L + 2J] \\
&= \big\lfloor \frac{M}{12} [ML + L - 2 + 2J] \big\rfloor.
\end{align*}
The proof is completed by checking that the above inequality is equivalent to the upper bound in the theorem.
\end{proof}

\section{Constructions of Optimal Multichannel CACs}
\label{sec:construction3}

In this section we give some constructions of optimal MC-CACs by combining some existing constructions of CACs designed for the single-channel case and a special class of combinatorial designs.

We first introduce some notation.
A CAC can be viewed as a special case of MC-CAC with $M=1$.
We will denote a CAC with length $L$ and weight $w$ by CAC$(L,w)$.
A scheduling pattern in a CAC$(L,w)$ can be represented as a set of time indices $\{t_1,t_2,\ldots, t_w\}$, where $t_1,\ldots, t_w$ are distinct elements in $\mathbb{Z}_L$.

A codeword in a CAC is called {\em equi-difference} if the elements in the corresponding scheduling pattern form an arithmetic progression in $\mathbb{Z}_L$.
A CAC is said to be {\em equi-difference} if all codewords are equi-difference.
An equi-difference codeword has the favorable property that there are at most $2(w-1)$ nonzero differences among the elements in the scheduling pattern.
The defining property of a CAC says that each element in $\{1,2,\ldots, L-1\}$ is equal to the difference of at most one pair of time indices in any scheduling pattern in the CAC.
An equi-difference CAC is said to be {\em tight} if for each element $\delta$ in $\{1,2,\ldots, L-1\}$, we can find a codeword and two time indices in the associated scheduling pattern such that their difference is equal to~$\delta$.

\medskip

\noindent {\bf Example 2.} The CAC with length $L=13$ consisting of scheduling patterns $\{0,1,2\}$, $\{0,3,6\}$ and $\{0,4,8\}$ is equi-difference and tight.

\medskip

We will need the following construction of CACs due to~\cite{Momihara07}.

\begin{theorem}[\cite{Momihara07}] Let $t$ be an odd prime such that $-1$ and $-3$ are quadratic non-residues mod~$t$. There exists a CAC of length $2t$ of weight 3 containing $(t-1)/2$ codewords.
\label{thm:Momihara}
\end{theorem}

\begin{definition}
\rm
Let $k$ and $L$ be two positive integers. A $k\times L$ matrix $X$, with each entry drawn from $\mathbb{Z}_L$, is called a {\em difference matrix} over $\mathbb{Z}_L$ if all elements in $\mathbb{Z}_L$ appear as the entries of the difference vector of any two distinct rows in~$X$.
\end{definition}

Difference matrix was first introduced by Bose and Bush in the context of orthogonal arrays~\cite{BB52}. A primary example of difference matrix is the mod-$p$ multiplication table for prime~$p$.  For prime number $L$ and positive integer $k\leq L$, we can construct a $k\times L$ difference matrix $X$ over $\mathbb{Z}_L$ by defining the $(i,j)$-entry as $ij \bmod p$, for $i=0,1,\ldots, k-1$ and $j=0,1,\ldots, L-1$. An example of difference matrix that arises from the mod-13 multiplication table is given in the next example.

\medskip

\noindent {\bf Example 3.} The following matrix is a $3\times 13$ difference matrix over $\mathbb{Z}_{13}$.
$$
\begin{array}{|c|c|c|c|c|c|c|c|c|c|c|c|c|} \hline
0&0&0&0&0&0&0&0&0&0&0&0&0 \\ \hline
0&1&2&3&4&5&6&7&8&9&10&11&12 \\ \hline
0&2&4&6&8&10&12&1&3&5&7&9&11 \\ \hline
\end{array}
$$

\medskip

We refer the readers to \cite{Evans02,Ge05} for more details on difference matrices with three rows and four rows.



The following definition is a special case of generalized Bhaskar Rao designs, in which the associated group is a finite cyclic group $\mathbb{Z}_L$~\cite{GBRD3}.

\begin{definition}
\rm
Given positive integers $M$, $L$ and $w$ with $w\leq M$, a {\em generalized Bhasker Rao Design} (GBRD) signed over $\mathbb{Z}_L$ is an $M\times b$ array $A$, where $b$ is an integer equal to $LM(M-1)/(w(w-1))$, satisfying the following three properties:

\ifdefined\SingleColumn
\renewcommand{\theenumi}{(\arabic{enumi})}
\else
\renewcommand{\theenumi}{(\arabic{enumi}}
\fi

\begin{enumerate}
\item each entry in the array $A$ is either empty or an element in $\mathbb{Z}_L$;
\item each column of $A$ contains exactly $w$ non-empty entries;
\item for each pair of distinct rows, say with distinct indices $i_1$ and $i_2$, there are exactly $L$ column indices, say $j_1$, $j_2,\ldots j_L$, such that $A(i_1,j_\ell)$ and $A(i_2,j_\ell)$ are both non-empty, for $\ell=1,2,\ldots, L$, and the differences $A(i_1,j_\ell) - A(i_2,j_\ell) \bmod L$, for $\ell=1,2,\ldots, L$, are equal to $0,1,\ldots, L-1$ after some permutation.
\end{enumerate}
\end{definition}

A GBRD reduces to a difference matrix when $M=w$, i.e., when all entries in the array are nonempty.

\medskip

\noindent {\bf Example 4.} A $4\times 20$ GBRD signed over $\mathbb{Z}_{10}$ is shown in Fig.~\ref{fig:GBRD}.

\medskip

\begin{figure*}[t]
$$
\begin{array}{|c|c|c|c|c|c|c|c|c|c|c|c|c|c|c|c|c|c|c|c|} \hline
0&0&0&0&0&0&0&0&0&0&0&0&0&0&0& & & & &\\ \hline
0&1&2&3&4&5&6&7&8&9& & & & & &0&0&0&0&0\\ \hline
0&2&4&6&8& & & & & &1&3&5&7&9&5&6&7&8&9\\ \hline
 & & & & &5&4&3&2&1&6&7&8&9&0&1&3&5&7&9\\ \hline
\end{array}
$$
\caption{A $4\times 20$ Generalized Bhaskar Rao Design over $\mathbb{Z}_{10}$}
\label{fig:GBRD}
\end{figure*}

Each column of a difference matrix or a GBRD can be regarded as a scheduling pattern in which the packets are transmitted in distinct channels. We illustrate the idea in the following example.

\medskip

\noindent {\bf Example 5.} We can regard the columns of the matrix in Example 3 as the scheduling patterns
$$
\{(0,0), (1, j \bmod 13), (2, 2j \bmod 13)  \}
$$
for $j=0,1,\ldots, 12$. We can combine these 13 codewords with 9 more codewords adopted from Example 2, namely,
\begin{align*}
&\{(i,0), (i,1), (i,2)\}, \\
&\{(i,0), (i,3), (i,6)\}, \text{ and}\\
&\{(i,0), (i,4), (i,8)\}, 
\end{align*}
for $i=0,1,2$, to form an MC-CAC$(3,13,3)$. The total number of codewords is $N=13+ 3\times 3 = 22$. The number of codewords attains the upper bound
$$
\Big\lfloor \frac{3}{12}(2\cdot 3\cdot 13 + 13 - 3)\Big\rfloor = 22
$$
in Theorem~\ref{thm:bound1}. This shows that the MC-CAC $(3,13,3)$ is optimal.

\medskip

The construction in Example 5 can be generalized. The next theorem gives a construction of MC-CACs.

\begin{theorem}\label{thm:optimal}
Let $M$, $L$ and $w$ be positive integers with $L \geq M\geq w$.
If there is an equi-difference and tight CAC$(L,w)$ $\mathcal{C}_0$ of weight $w=3$ (resp. $w=4$), and an $M \times b$ GBRD $X$ over $\mathbb{Z}_L$ in which each column contains exactly $w=3$ (resp. $w=4$) nonempty entries, then we can construct an optimal MC-CAC$(M,L,w)$.
\end{theorem}

\begin{proof}
The GBRD $X$ has size $M \times LM(M-1)/(w(w-1))$. This gives rise to $LM(M-1)/(w(w-1))$ codewords of type (ii).
We then add codewords of type (i) obtained from $\mathcal{C}_0$.
Recall that an equi-difference codeword of weight 3 (resp. weight 4) in a CAC can be represented by a subset of $\mathbb{Z}_L$ in the form $\{0,a,2a\}$ (resp. $\{0,a,2a,3a\}$). We say that the codeword is generated by $a$, or $a$ is a {\em generator} of the codeword. For each codeword in $\mathcal{C}_0$ with generator $a$, we put $M$ scheduling patterns
$$
\{ (i,0), (i,a), \ldots, (i,(w-1)a)\}
$$
for $i=1,2,\ldots, M$, in the resulting MC-CAC.

In the followings, we say that a codeword in a CAC is {\em exceptional} if there are strictly less than $2(w-1)$ differences in the codeword.

We first consider weight $w=3$. There are at most four differences from an equi-difference codeword $\{0,a,2a\}$, namely $\pm a$ and $\pm 2a$. We have two possible exceptional codewords. The first one is generated by $L/3$ when $L$ is divisible by 3, and the second one is generated by $L/4$ when $L$ is divisible by 4. We distinguish four cases in the followings.

\begin{enumerate}[(a)]
\item Suppose that there is no exceptional codeword in $\mathcal{C}_0$, i.e., each codeword contributes exactly four nonzero differences.
By the tightness assumption, the number of codewords in $\mathcal{C}_0$ is $(L-1)/4$.
Hence $(L-1)$ is divisible by 4 and $\gcd(L,4)=1$. Also, $L$ cannot be divisible by~3. If $L$ is divisible by 3, then a tight equi-difference CAC must contain the exceptional codeword $\{0,L/3, 2L/3\}$, and it violates the assumption that there is no exceptional codeword.
This proves that $\gcd(L,12)=1$ in this case.
By putting $(L-1)/4$ scheduling patterns in each channel, the total number of scheduling patterns is
\begin{align*}
\frac{M(L-1)}{4} + \frac{LM(M-1)}{6} = \frac{M}{12} (2ML + L - 3).
\end{align*}
This attains the upper bound in Theorem~\ref{thm:bound1} when $L$ and 12 are relatively prime.

\item Suppose $L$ is divisible by 3 and $\{0,L/3, 2L/3\}$ is the only exceptional codeword in $\mathcal{C}_0$.
Since $C_0$ is tight, $L-3$ must be divisible by 4, and thus $L = 3 \mod 12$.
The number of scheduling patterns in the resulting MC-CAC is
\begin{align*}
M+ \frac{M(L-3)}{4} + \frac{LM(M-1)}{6} = \frac{M}{12} (2ML + L +3).
\end{align*}
The upper bound in Theorem~\ref{thm:bound1} with $L = 3 \bmod 12$ is met with equality.

\item Suppose $L$ is divisible by 4 and $\{0,L/4, L/2\}$ is the only exceptional codeword in~$\mathcal{C}_0$.
By the argument similar to the previous case, we can show that when $\gcd(L,12)=4$, and there are
$$
 \frac{M}{12}(2ML+L)
$$
scheduling patterns in the resulting MC-CAC, which is optimal by Theorem~\ref{thm:bound1}.

\item Suppose CAC $\mathcal{C}_0$ contains both exceptional codewords $\{0,L/3, 2L/3\}$ and $\{0,L/4, L/2\}$.
Then $L$ is divisible by 12, and we have
$$
 2M + \frac{M(L-6)}{4} + \frac{LM(M-1)}{6} = \frac{M}{12}(2ML+L+6)
$$
scheduling patterns in the resulting MC-CAC.
This is optimal since the quantity attains the upper bound in Theorem~\ref{thm:bound1} when $L =0 \bmod 12$.
\end{enumerate}

We next consider weight $w=4$.
An equi-difference codeword $\{0,a,2a,3a\}$ has at most six differences in the set of differences.
The number of differences is strictly less than six if $a=L/5$ when $L$ is divisible by 5, and when $a=L/4$ when $L$ is divisible by~4.
So, there are at most two exceptional codewords in $\mathcal{C}_0$.
We divide the proof into four cases.

\begin{enumerate}[(a)]
\item Suppose that there is no exceptional codeword in $\mathcal{C}_0$.
We have $(L-1)/6$ codewords in $\mathcal{C}_0$, and thus $L$ is not divisible by 2 or 3.
We also have $L$ not divisible by 5. Otherwise, because $\mathcal{C}_0$ is assumed to be tight, it must contain the codeword $\{0,L/5, 2L/5, 3L/5\}$.
Hence, we have $\gcd(L,60)=1$.
There are precisely $(L-1)/6$ codewords in~$\mathcal{C}_0$.
Therefore, total number of scheduling patterns in the resulting MC-CAC is
\begin{align*}
\frac{M(L-1)}{6} + \frac{LM(M-1)}{12} = \frac{M}{12}(ML+L-2).
\end{align*}
Comparing with the upper bound in Theorem~\ref{thm:b4} when $\gcd(L,60)=1$, we see that the largest value of $N$ is attained.

\item Suppose $L$ is divisible by 4 and $\{0,L/4, L/2, 3L/4\}$ is the only exceptional codeword in~$\mathcal{C}_0$.
Since $\mathcal{C}_0$ is tight, we have $L=6k+4$ for some integer $k$.
In particular, $L$ is not divisible by~3. Furthermore, $L$ cannot be divisible by 5 neither.
Otherwise $\mathcal{C}_0$ would contain the other exceptional codeword $\{0,L/5, 2L/5, 3L/5\}$, contradicting the assumption that $\{0,L/4, L/2, 3L/4\}$ is the only exceptional codeword in~$\mathcal{C}_0$.
The total number of scheduling patterns in the resulting MC-CAC is
$$
M+\frac{M(L-4)}{6} + \frac{LM(M-1)}{12} = \frac{M}{12}(ML+L+4),
$$
which matches the upper bound in Theorem~\ref{thm:b4} when $\gcd(L,60)=4$.

\item Suppose $L$ is divisible by 5 and $\{0,L/5, 2L/5, 3L/5\}$ is the only exceptional codeword in~$\mathcal{C}_0$.
By the assumption that $\mathcal{C}_0$ is tight, $(L-5)/6$ is an integer, and hence $L$ is not divisible by 6.
This yields $\gcd(L,60)=5$.
We have
$$
M+\frac{M(L-5)}{6} + \frac{LM(M-1)}{12} = \frac{M}{12}(ML+L+2)
$$
scheduling patterns in the resulting MC-CAC, and this is optimal by the upper bound in Theorem~\ref{thm:b4} when $\gcd(L,60)=5$.

\item Finally, suppose $L$ is divisible by 20 and the CAC $\mathcal{C}_0$ contains both  $\{0,L/4, L/2, 3L/4\}$ and $\{0, L/5, 2L/5, 3L/5\}$ as codewords.
Then $\mathcal{C}_0$ contains $2+(L-8)/6$ codewords. Since $(L-8)$ is a multiple of 6, $L$ cannot be divisible by~3.
We get $\gcd(L,60)=20$ in this case.
The total number of codewords in the resulting MC-CAC is
$$
2M+\frac{M(L-5)}{6} + \frac{LM(M-1)}{12} = \frac{M}{12}(ML+L+8),
$$
which attains the upper bound in Theorem~\ref{thm:b4} when $\gcd(L,60)=20$.
\end{enumerate}
\end{proof}

For any positive odd integer $n$, let $e_n$ be the smallest exponent $e\geq 1$ so that $2^e\equiv 1$ mod $L$.
The value $e_n$ is called the multiplicative order of 2 mod $n$, see~\cite{OEIS_multiplicative}.
It was shown in~\cite[Theorem 4]{HLS10} that an equi-difference and tight CAC$(L,3)$ exists, whenever each prime factor $p$ of the length $L$ satisfies
\begin{equation}\label{eq:tight_condition}
p\equiv 5 \text{ mod } 8 \text{ or } p\equiv 1 \text{ mod } 8 \text{ and }4|e_p.
\end{equation}
In this case, $\gcd(L,6)=1$.
Moreover, a GBRD over $\mathbb{Z}_L$ exists if $L$ is odd, $M\geq 3$, $w=3$ and $LM(M-1)\equiv 0$ mod $6$~\cite[Theorem 2 and Lemma 20]{GBRD3}.
We observe that when $\gcd(L,6)=1$, the condition $LM(M-1)\equiv 0$ mod $6$ implies $M\equiv 0,1$ mod $3$.
Hence, as an application of Theorem~\ref{thm:optimal}, we can derive optimal CAC$(M,L,3)$s, where $M\equiv 0,1$ mod $3$ and each prime factor $p$ of $L$ satisfies \eqref{eq:tight_condition}.
The first few feasible lengths $L$ are $$5,13,17,29,37,41,53,61,65,85,97,101.$$

In the rest of this section, we give an infinite family of optimal MC-CACs whose codeword lengths are even.

\medskip

\noindent {\bf Construction}. For any positive integer $t$, one can construct a $4\times 4t$  GBRD signed over $\mathbb{Z}_{2t}$ in which every column has exactly 3 nonempty entries~\cite[Example 48]{GBRD3}. This gives the following $4t$ scheduling patterns,
\begin{gather*}
\{(0, 0), (1, j), (2, 2j)\},\
\{(0, 0), (1, t+j), (3, t-j)\}, \\
\{(0, 0), (2, 2j+1), (3, t+j+1)\},\\
\{(1, 0), (2, t+j), (3, 2j+1)\},
\end{gather*}
for $j=0,1,\ldots, t-1$. We take $t$ to be a prime number larger than 3 such that $-1$ and $-3$ are quadratic non-residues mod $t$.
We apply Theorem~\ref{thm:Momihara} to obtain a CAC$(2t,3)$ that contains $(t-1)/2$ codewords. Repeat the CAC 4 times, once for each of the 4 channels, we get $2(t-1)$ codewords of type (i). In total, we have
$ 4t+2(t-1) = 6t-2$ codewords. Note that $t\equiv 1$ or $5 \mod 6$ since it is an odd prime number. This is an optimal MC-CAC$(4,2t,3)$, because it meets the upper bound on the code size in Theorem~\ref{thm:bound1},
$$
A(4,2t,3) \leq \Big\lfloor \frac{4}{12}(2\cdot 4 \cdot (2t) + 2t - 6) \Big\rfloor = \big\lfloor (18t-6)/3 \big\rfloor = 6t-2.
$$

There are infinitely many primes such that $-1$ and $-3$ are both quadratic non-residues (see e.g.~\cite{SW10}).
Hence, the above construction gives an infinite family of optimal MC-CACs of weight $w=3$ for $M=4$ channels.

\medskip

\noindent {\bf Example 6.} We can convert the GBRD  in Example~4 (see Fig.~\ref{fig:GBRD}) to codewords in an MC-CAC$(4,10,3)$. The columns in the GBRD are associated with the following scheduling patterns
\begin{align*}
&\{(0,0), (1,0), (2,0)\},
\{(0,0), (1,1), (2,2)\}, \\
&\{(0,0), (1,2), (2,4)\},
\{(0,0), (1,3), (2,6)\}, \\
&\{(0,0), (1,4), (2,8)\},
\{(0,0), (1,5), (3,5)\}, \\
&\{(0,0), (1,6), (3,4)\},
\{(0,0), (1,7), (3,3)\}, \\
&\{(0,0), (1,8), (3,2)\},
\{(0,0), (1,9), (3,1)\}, \\
&\{(0,0), (2,1), (3,6)\},
\{(0,0), (2,3), (3,7)\}, \\
&\{(0,0), (2,5), (3,8)\},
\{(0,0), (2,7), (3,9)\}, \\
&\{(0,0), (2,9), (3,0)\},
\{(1,0), (2,5), (3,1)\}, \\
&\{(1,0), (2,6), (3,3)\},
\{(1,0), (2,7), (3,5)\}, \\
&\{(1,0), (2,8), (3,7)\},
\{(1,0), (2,9), (3,9)\}.
\end{align*}
They correspond to 20 codewords in the MC-CAC$(4,10,3)$. In channel $i$, we insert two more codewords,
$$
\{(i,0),(i,1),(i,2) \}\ \text{ and }
\{(i,0),(i,3),(i,6) \},
$$
for $i = 0,1,2,3$. We thus obtain an MC-CAC$(4,10,3)$ consisting of 28 codewords. By comparing with the upper bound in Theorem~\ref{thm:bound1},
$$
A(4,10,3) \leq \Big\lfloor \frac{4}{12}(2\cdot 4 \cdot 10 + 10 - 6)\Big\rfloor = \big\lfloor 84/3 \big\rfloor= 28,
$$
we see that this is an optimal MC-CAC.

\section{Concluding Remarks}

We study MC-CACs when several orthogonal channels can be used for transmissions.
By viewing the problem as a combinatorial packing problem, we show that the total number of supported source nodes increases in the order of $M^2L/6 +O(M)$ for weight 3, and in the order of $M^2L/12 +O(M)$ for weight 4, where $M$ denotes the number of channels, and $L$ denotes the codeword length.
We conjecture that the maximal number of codewords in an MC-CAC$(M,L,w)$, for each fixed weight $w$, has order $O(M^2L)$.
Nevertheless, the number of types involved in the analysis is equal to the number of ways we can partition $M$ into non-negative integers, which is known as the partition function $p(M)$~\cite{OEIS_partition}.
The analysis becomes more complicated when $M$ increases.

The construction given in Theorem~\ref{thm:optimal} also works for weight $w \geq 5$, provided that an equi-difference and tight CAC of weight $w$ and a GBRD with weight $w$ are available.
Unfortunately, existing results for CACs of weight $w\geq 5$ and GBRDs with weight $w\geq 5$ are relatively scarce in comparison to those with weights 3 and~4. It is not easy to extend Theorem~\ref{thm:optimal} to weight $w\geq 5$.
Nonetheless, if we can find a CAC, which may be sub-optimal, and a GBRD with weight $w\geq 5$ so that the construction is applicable, we can indeed construct an MC-CAC.
But we cannot guarantee that this MC-CAC is optimal as we do not have an upper bound of code size for weight $w\geq 5$ yet.
The determination of the optimal code size $A(M,L,w)$ for $w\geq 5$ is an interesting direction for further research.

We focus on an uplink scenario in this paper. The idea of using scheduling patterns in communications can also be considered in ad hoc networks.
Nodes in an ad hoc network may want to broadcast data to the other nodes, or want to send different data to different nodes, without any central coordinator.
A design of scheduling patterns for unicast can be found in~\cite{LSW20}.

The multichannel conflict-avoiding codes defined in this paper allows simultaneously transmitting two packets in two different channels. In a more practical setting considered in~\cite{CLW19}, it is assumed that in each time slot each source node can only pick one channel and send one packet in the chosen channel. Hence, the set of differences $D_S(m_1,m_2)$ should not contain the zero element in $\mathbb{Z}_L$ whenever $m_1\neq m_2$. The bounds derived in Section~\ref{sec:bound3} and \ref{sec:bound4} can be easily adapted to this case. For example, for $w=3$, we only need to replace the right-hand side of~\eqref{eq:LP1a} by $\binom{M}{2}(L-1)$ and repeat the derivation. We state the upper bound on code size with this additional and practical assumption as follows. For weight $w=3$, the number of codewords is upper bounded by
$$
N \leq \Big\lfloor \frac{M}{12}((2M+1)(L-1)+3\mathbf{1}_{\{L \in 2 \mathbb{Z}\}}+6\mathbf{1}_{\{L \in 3 \mathbb{Z}\}}) \Big\rfloor.
$$
For weight $w=4$, the number of codewords $N$ is upper bounded by
\begin{align*}
N &= \Big\lfloor \frac{M}{12}\big[(M+1)(L-1) +2J\big]  \Big\rfloor
\end{align*}
where $J$ is defined in \eqref{eq:J}. 
The study of MC-CACs with the restriction of sending at most one packet in a time slot is another meaningful research direction.

\end{document}